\documentclass[submission,copyright,creativecommons]{eptcs}
\providecommand{\event}{CL\&C 2010} 
\usepackage{breakurl}             

\usepackage[utf8]{inputenc}
\usepackage[T1]{fontenc}

\usepackage{amsmath, amssymb}
\usepackage{amsthm} 
\usepackage{paralist}
\usepackage{subfigure}

\usepackage{mathpartir}
\def \TirNameStyle #1{\small #1}

\newtheorem{defi}{Definition}
\newtheorem{theo}{Theorem}
\newtheorem{hypo}{Hypothesis}

\newcommand{\tmu}{\tilde{\mu}}
\newcommand{\lmmt}{\overline{\lambda}\mu\tmu}
\newcommand{\tez}{\vdash}
\newcommand{\imply}{\Rightarrow}
\newcommand{\fail}{\mathfrak{f}}
\newcommand{\typed}{\rhd}
\newcommand{\fv}{\mathcal{FV}}
\newcommand{\comma}{,}
\newcommand{\fox}{\mathsf{x}}
\newcommand{\foy}{\mathsf{y}}
\newcommand{\fot}{\mathsf{t}}
\newcommand{\fou}{\mathsf{u}}

\newcommand{\lmmtR}{\lmmt_\mathcal{R}}
\newcommand{\nfp}{\!\downarrow_p}
\newcommand{\nfe}{\!\downarrow_e}

\usepackage{url}

\title{Superdeduction in $\lmmt$}
\author{Cl\'ement Houtmann
\institute{INRIA Saclay - \^Ile de France and LIX/\'Ecole Polytechnique, 91128 Palaiseau Cedex}
\email{Clement.Houtmann@inria.fr}
}
\def\authorrunning{C. Houtmann}
\def\titlerunning{Superdeduction in $\lmmt$}
\begin{document}
\maketitle

\begin{abstract}
Superdeduction is a method specially designed to ease the use of first-order
theories in predicate logic. The theory is used to enrich the deduction system
with new deduction rules in a systematic, correct and complete way.
A proof-term language and a cut-elimination reduction already exist for
superdeduction, both based on Christian Urban's work on classical sequent
calculus.  However the computational content of Christian Urban's calculus is
not directly related to the ($\lambda$-calculus based) Curry-Howard
correspondence. In contrast the $\lmmt$-calculus is a $\lambda$-calculus for
classical sequent calculus.
This short paper is a first step towards a further exploration of the
computational content of superdeduction proofs, for we extend the
$\lmmt$-calculus in order to obtain a proofterm langage together with a
cut-elimination reduction for superdeduction. We also prove strong
normalisation for this extension of the $\lmmt$-calculus.

\end{abstract}

\section{Introduction}

\paragraph{Superdeduction} is an extension of predicate logic designed to ease
the use of first-order theories by enriching a deduction system with new
deduction rules computed from the theory. Once the theory is presented as a
rewrite system, the translation into a set of custom (super)deduction rules is
fully systematic. Superdeduction systems \cite{BraunerHK07} are usually
constructed on top of the classical sequent calculus LK which is described in
Figure \ref{fig:LK}.
\begin{figure}
\[
\begin{array}{c}

\inferrule*[left=Ax]
{ }
{\Gamma, \varphi \tez \varphi, \Delta}

\hspace{5mm}

\inferrule*[left=Cut]
{\Gamma \tez \varphi, \Delta ~~~
\Gamma, \varphi \tez \Delta}
{\Gamma \tez \Delta}

\hspace{5mm}
\inferrule*[left=ContrR]
{\Gamma \tez \varphi, \varphi, \Delta}
{\Gamma \tez \varphi, \Delta}

\hspace{5mm}

\inferrule*[left=ContrL]
{\Gamma, \varphi, \varphi \tez \Delta}
{\Gamma, \varphi \tez \Delta}

\\[2mm]

\inferrule*[left=$\bot$R]
{\Gamma \tez \Delta}
{\Gamma \tez \bot, \Delta}

\hspace{5mm}

\inferrule*[left=$\bot$L]
{ }
{\Gamma, \bot \tez \Delta}

\hspace{5mm}
\inferrule*[left=$\top$R]
{ }
{\Gamma \tez \top, \Delta}

\hspace{5mm}

\inferrule*[left=$\top$L]
{\Gamma \tez \Delta}
{\Gamma, \top \tez \Delta}

\\[2mm]

\inferrule*[left=$\land$R]
{\Gamma \tez \varphi_1, \Delta \\
\Gamma \tez \varphi_2, \Delta}
{\Gamma \tez \varphi_1 \land \varphi_2, \Delta}

\hspace{5mm}

\inferrule*[left=$\land$L]
{\Gamma, \varphi_1, \varphi_2 \tez \Delta}
{\Gamma, \varphi_1 \land \varphi_2 \tez \Delta}

\hspace{5mm}

\inferrule*[left=$\imply$R]
{\Gamma, \varphi_1 \tez \varphi_2, \Delta}
{\Gamma \tez \varphi_1 \imply \varphi_2, \Delta}

\\[2mm]

\inferrule*[left=$\lor$R]
{\Gamma \tez \varphi_1, \varphi_2, \Delta}
{\Gamma \tez \varphi_1 \lor \varphi_2, \Delta}

\hspace{5mm}

\inferrule*[left=$\lor$L]
{\Gamma, \varphi_1 \tez \Delta \\
\Gamma, \varphi_2 \tez \Delta}
{\Gamma, \varphi_1 \lor \varphi_2 \tez \Delta}

\hspace{5mm}

\inferrule*[left=$\imply$L]
{\Gamma \tez \varphi_1, \Delta \\
\Gamma, \varphi_2 \tez \Delta}
{\Gamma, \varphi_1 \imply \varphi_2 \tez \Delta}

\\[2mm]

\inferrule*[left=$\forall$R,right={$\fox \notin
\fv(\Gamma\comma\Delta)$}]
{\Gamma \tez \varphi, \Delta}
{\Gamma \tez \forall \fox . \varphi, \Delta}

\hspace{5mm}

\inferrule*[left=$\forall$L]
{\Gamma, \varphi[\fot \slash \fox] \tez \Delta}
{\Gamma, \forall \fox. \varphi \tez \Delta}

\hspace{5mm}

\inferrule*[left=$\exists$R]
{\Gamma \tez \varphi[\fot \slash \fox], \Delta}
{\Gamma \tez \exists \fox. \varphi, \Delta}

\hspace{5mm}

\inferrule*[left=$\exists$L,right={$ \fox \notin \fv(\Gamma\comma\Delta)$}]
{\Gamma, \varphi \tez \Delta}
{\Gamma, \exists \fox.\varphi \tez \Delta}

\end{array}
\]
\caption{Classical Sequent Calculus LK}\label{fig:LK}
\end{figure}
New deduction rules are computed from a theory presented as a set of
\emph{proposition rewrite rules}, \mbox{i.e.} rewrite rules of the form $P
\rightarrow \varphi$ where $P$ is some atomic formula. Such rewrite rules
actually stand for equivalences $\forall \overline{\fox}. (P \Leftrightarrow \varphi)$
where $\overline{\fox}$ represents the free variables of $P$. The computation of
custom inferences for the proposition rewrite rule $P \rightarrow \varphi$
goes as follows.  On the right, the algorithm decomposes (bottom-up) the
sequent $\tez \varphi$ using LK$\backslash
\{\mbox{Cut},\mbox{ContrR},\mbox{ContrL}\}$ (non-deterministically) until it
reaches a sequence of atomic sequents\footnote{\emph{i.e.} sequents containing
only atomic formul\ae{}} $(\Gamma_i \tez \Delta_i)_{1\leqslant i \leqslant
n}$. During this decomposition, each application of $\exists$L and $\forall$R
corresponds to a side condition $\fox \notin \fv(\Upsilon)$ for some
first-order variable $\fox$ and for some list of formula $\Upsilon$. This
particular decomposition of $\tez \varphi$ then leads to the inference rule
$$
\inferrule*[right=$C$]
{(\Gamma, \Gamma_i \tez \Delta_i, \Delta)_{1\leqslant i\leqslant n}}
{\Gamma \tez P, \Delta}
$$
for introducing $P$ on the right where $C$ is the conjunction of the side
conditions.
On the left, the algorithm similarly decomposes $\varphi \tez$ until it reaches a
sequence of atomic sequents ${(\Gamma'_j \tez \Delta'_j)_{1\leqslant
j \leqslant m}}$ and a conjunction of side conditions $C'$ yielding similarly
the inference rule
$$
\inferrule*[right=$C'$]
{(\Gamma, \Gamma'_j \tez \Delta'_j, \Delta)_{1 \leqslant j \leqslant m}}
{\Gamma, P \tez \Delta}~.
$$
As remarked in \cite{Houtmann08}, this non-deterministic algorithm may return
several inference rules for introducing $P$ respectively on the right or on
the left. One must add all the possible inference rules in order to obtain a
complete superdeduction system.
\begin{defi}[Superdeduction systems \cite{BraunerHK07}]
If $\mathcal{R}$ is a set of proposition rewrite rules, the superdeduction
system associated with $\mathcal{R}$ is obtained by adding to LK all
the inferences which can be computed from the elements of $\mathcal{R}$.
\end{defi}

The paradigmatic example for superdeduction is the system associated with the
proposition rewrite rule $A \subseteq B ~~\rightarrow~~ \forall \fox. (\fox \in A
\imply \fox \in B)$ which defines the inclusion predicate $\subseteq$. This
rewrite rule yields inference rules
$$
\inferrule*[right=$\fox \notin \fv(\Gamma\comma\Delta)$]
{\Gamma, \fox \in A \tez \fox \in B, \Delta}
{\Gamma \tez A \subseteq B, \Delta}
\hspace{5mm}\mbox{and}\hspace{5mm}
\inferrule*[]
{\Gamma, \fot \in B \tez \Delta \\
\Gamma \tez \fot \in A, \Delta}
{\Gamma, A \subseteq B \tez \Delta}~.
$$

As demonstrated in \cite{BraunerHK07}, superdeduction systems are always sound
\emph{w.r.t.} predicate logic. Completeness is ensured whenever right-hand
sides of proposition rewrite rules do not alternate
quantifiers\footnote{Formul\ae{} such as $(\forall \fox. \varphi) \land (\exists
\foy.  \psi)$ are allowed.}. Cut-elimination is more difficult to obtain: several
counterexamples are displayed in \cite{TheseClement}. We have proved in
\cite{Houtmann08} that whenever right-hand sides of proposition rewrite rules
do not contain universal quantifiers and existential quantifiers at the same
time\footnote{Formul\ae{} such as $(\forall \fox. \varphi) \land (\exists \foy.
\psi)$ are \emph{not} allowed.}, cut-elimination in superdeduction is
equivalent to cut-elimination in \emph{deduction modulo} (another formalism
which removes computational arguments from proofs by reasoning modulo
rewriting on propositions \cite{TPM-DHK-JAR-03}).

In the original paper introducing superdeduction \cite{BraunerHK07}, a
proof-term language and a cut-elimination reduction are defined for
superdeduction, both based on Christian Urban's work on classical sequent
calculus \cite{UrbanThese2000}. The reduction is proved to be strongly
normalising on well-typed terms when the set of proposition rewrite rules
$\mathcal{R}$ satisfies the following hypothesis.
\begin{hypo}\label{hypo}
The rewriting relation associated with $\mathcal{R}$ is weakly normalising and
confluent and no first-order function symbol appears in the left-hand sides of
proposition rewrite rules of $\mathcal{R}$.
\end{hypo}

The computational content of Christian Urban's calculus is not directly
related to the (functional) Curry-Howard correspondence whereas the
$\lmmt$-calculus \cite{CurienH00} is a $\lambda$-calculus for sequent
calculus. In order to explore the computational content of superdeduction
inferences, we will define in Section \ref{sec:lmmt+} an extension of the
$\lmmt$-calculus for superdeduction systems and prove the same strong
normalisation result using Hypothesis \ref{hypo}. But before doing so, let us
recall the definition of the $\lmmt$-calculus.

\paragraph{The $\lmmt$-calculus} is defined as follows. In order to avoid
confusion between first-order variables and $\lmmt$ variables, we will use
sans-serif symbols for first-order variables ($\fox,\foy\dots$) and
first-order terms ($\fot,\fou\dots$). Commands,
terms and environments are respectively defined by the grammar in Figure
\ref{subfig:lmmtGrammar}.
The type system is described in Figure \ref{subfig:lmmtTyping}.
Reduction rules are depicted in Figure \ref{subfig:lmmtReductions}.
We have added a constant environment $\fail$ in order to realise falsity. We
also have added constructions $\lambda \fox. \pi$ and $\fot \cdot e$ in order
to realise universal quantifications respectively on the right and on the
left. Implication, universal quantification and falsity are sufficient to
express all the connectives in LK.
\begin{figure}
\centering
\hfill
\subfigure[Grammar]{
\centering
$
\begin{array}{r@{~::=~}l@{\quad}l}
c & \langle \pi | e \rangle & \mbox{(commands)}\\
\pi & x \mid \lambda x. \pi \mid \mu \alpha. c \mid \lambda \fox. \pi &
\mbox{(terms)}\\
e & \alpha \mid \pi \cdot e \mid \tmu x . c \mid \fot \cdot e \mid \fail &
\mbox{(environments)}\\[2mm]
\end{array}
$
\label{subfig:lmmtGrammar}}
\hfill
\subfigure[Reduction]{
\centering
$
\begin{array}{l@{~\rightarrow~}l}
\langle \lambda x. \pi | \pi' \cdot e \rangle &
\langle \pi[\pi'\slash x] | e \rangle \\
\langle \mu \alpha. c | e \rangle &
c[e\slash \alpha] \\
\langle \pi | \tmu x . c \rangle &
c[\pi \slash x]\\
\langle \lambda \fox. \pi | \fot \cdot e \rangle &
\langle \pi[\fot\slash\fox] | e \rangle
\\[2mm]
\end{array}
$
\label{subfig:lmmtReductions}}\hfill~\\
\subfigure[Type System]{
\centering
$
\begin{array}{c}
\inferrule*[]
{ }
{\Gamma, x:A \tez x:A \mid \Delta}
\hspace{1cm}
\inferrule*[]
{ }
{\Gamma \mid \alpha : A \tez \alpha : A, \Delta}
\hspace{1cm}
\inferrule*[]
{\Gamma, x : A \tez \pi : B \mid \Delta}
{\Gamma \tez \lambda x. \pi : A \imply B \mid \Delta}
\\[2mm]
\inferrule*[]
{\Gamma \tez \pi : A \mid \Delta \\ \Gamma \mid e : B \tez \Delta}
{\Gamma \mid \pi \cdot e : A \imply B \tez \Delta}
\hspace{1cm}
\inferrule*[left=Cut]
{\Gamma \tez \pi : A \mid \Delta \\ \Gamma \mid e : A \tez \Delta}
{\langle \pi | e \rangle \typed \Gamma \tez \Delta}
\hspace{1cm}
\inferrule*[]
{ }
{\Gamma \mid \fail : \bot \tez \Delta}
\\[2mm]
\inferrule*[]
{c \typed \Gamma \tez \alpha : A, \Delta}
{\Gamma \tez \mu \alpha. c : A \mid \Delta}
\hspace{8mm}
\inferrule*[]
{c \typed \Gamma, x : A \tez \Delta}
{\Gamma \mid \tmu x. c : A \tez \Delta}
\hspace{8mm}
\inferrule*[right=$\fox \notin \fv(\Gamma\comma\Delta)$]
{\Gamma \tez \pi : A \mid \Delta}
{\Gamma \tez \lambda \fox. \pi : \forall \fox . A \mid \Delta}
\hspace{8mm}
\inferrule*[]
{\Gamma \mid e : A[\fot\slash \fox] \tez \Delta}
{\Gamma \mid \fot \cdot e : \forall \fox. A \tez \Delta}
\\[2mm]
\end{array}
$
\label{subfig:lmmtTyping}
}
\caption{The $\lmmt$-calculus}\label{fig:lmmt}
\end{figure}
The typing rules
$$
\inferrule*[left=FocusR]
{\Gamma \tez \pi : A \mid \Delta}
{\langle \pi | \alpha \rangle \typed \Gamma \tez \alpha : A, \Delta}
\hspace{1cm}\mbox{and}\hspace{1cm}
\inferrule*[left=FocusL]
{\Gamma\mid e : A \tez \Delta}
{\langle x | e \rangle \typed \Gamma, x : A \tez \Delta}
$$
are admissible in the type system of Figure \ref{subfig:lmmtTyping}. Replacing
the Cut rule by FocusR and FocusL yields a type system that we will call
\emph{cut-free $\lmmt$}. It is obviously not equivalent to the original type
system in Figure \ref{subfig:lmmtTyping}. The reduction relation defined in
Figure \ref{subfig:lmmtReductions} is strongly normalising on well-typed terms
as demonstrated in \cite{Polonovski04}.

\paragraph{\bf Notations.} Sequences $(a_i)_{1\leqslant i \leqslant n}$ may be
denoted $(a_i)_i$ or just $\bar{a}$ when the upper bound $n$ can be retrieved
from the context (or is irrelevant). Both notations may even be combined:
$(\bar{a}_i)_i$ represents a sequence of sequences $((a_{j,i})_{1\leqslant j
\leqslant m_i})_{1\leqslant i \leqslant n}$. Finally if $\Gamma = (A_i)_i$ and
$\bar{x} = (x_i)_i$ are respectively a sequence of $n$ formul\ae{} and a
sequence of $n$ variables, then $\bar{x} : \Gamma$ denotes the (typed) context
$x_1 : A_1, x_2 : A_2 \dots$

\section{Extending $\lmmt$}\label{sec:lmmt+}

In the paper introducing superdeduction \cite{BraunerHK07}, Christian Urban's
calculus is presented as a better choice than the $\lmmt$-calculus for a basis
of a proofterm language for superdeduction. In this section, we refute this
claim and demonstrate that the $\lmmt$-calculus is as suitable as Christian
Urban's calculus. Such an extension is a first step towards a
Curry-Howard based computational interpretation of superdeduction, since the
$\lmmt$-calculus relates directly to the $\lambda$-calculus. An inaccuracy of
the original paper \cite{BraunerHK07} is also corrected in the process.
The extension of the $\lmmt$-calculus that we
will present corrects this mistake. The imprecision concerns first-order
quantifications. Indeed a superdeduction inference represents an open
derivation which may contain several quantifier destructions. The structure
organizing these destructions is essential to the definition of the underlying
cut-elimination mechanisms. For instance a sequence $\forall\exists$ on the
right corresponds to the creation of an eigenvariable, say $\fox$, followed by
an instantiation by some first-order term, say $\fot$, which may contain
$\fox$ as a free variable. A sequence $\exists\forall$ on the right
corresponds to an instantiation by some first-order term, say $\fot$, followed
by the creation of an eigenvariable, say $\fox$. In this latter case, $\fot$
is not allowed to contain $\fox$ as a free variable. This distinction is
completely erased in the syntax of the original extension \cite{BraunerHK07}.
It results in an imprecision of the scope of eigenvariables in extended
proofterms: the scope is not explicit in the syntax.
In our extension of the $\lmmt$-calculus, this syntactical imprecision is
corrected by introducing a notion of \emph{trace} which represents the correct
syntax for a precise syntactical representation of the scopes of eigenvariables
in extended proofterms. Then we present a correct cut-elimination
procedure by introducing a notion of \emph{interpretation} for the constructs
of the extended $\lmmt$-calculus relating such constructs to $\lmmt$ proofterms in a
correct way. At the end of the section, a pathological example is depicted to
illustrate the imprecision of the original extension \cite{BraunerHK07} and the
correction of the present extension.

First, let us consider any derivation in LK, potentially unfinished,
\emph{i.e.} with leaves that remain unproven. Since such a derivation is a
tree, there exists a natural partial order on its inferences: an inference
precedes another if the former is placed under the latter. Such a partial order
can easily be extended into a total order (in a non-deterministic way).
Considering only instances of $\forall$R, $\forall$L, $\exists$R and
$\exists$L, such a total order returns a list $L$ of such instances. Each
instance of $\forall$R or $\exists$L corresponds to the use of an
eigenvariable, say $\fox$. Such a use will be denoted $\fox?$. Each instance of
$\forall$L or $\exists$R corresponds to the instantiation of some first-order
variable by a first-order term, say $\fot$. Such a use will be denoted $\fot!$.
The list $L$ becomes a list whose elements are either of the form $\fox?$ or of
the form $\fot!$. Such a list is called a \emph{trace} for the derivation.

Let us consider a proposition rewrite rule $r: P \rightarrow \varphi$ leading
to the superdeduction inferences
$$
\inferrule*[right=$C$]
{(\Gamma, \Gamma_i \tez \Delta_i, \Delta)_i}
{\Gamma \tez P, \Delta}
\hspace{5mm}\mbox{and}\hspace{5mm}
\inferrule*[right=$C'$]
{(\Gamma, \Gamma'_j \tez \Delta'_j, \Delta)_j}
{\Gamma, P \tez \Delta}~.
$$
Let us consider the first one. Since it is derived from inferences of LK,
there exists a derivation of $\tez \varphi$ with open leaves $(\Gamma_i \tez
\Delta_i)_i$ in LK \cite[Property 6.1.3]{TheseClement}. Let $L$ be a trace for
this derivation. Then the superdeduction inference introducing $P$ on the right
is turned into the typing rule
$$
\inferrule*[left=rR,right=$C$]
{(c_i \typed \Gamma, \overline{x}_i:\Gamma_i \tez \overline{\alpha}_i:\Delta_i, \Delta)_i}
{\Gamma \tez r(L,(\mu_i(\overline{x}_i,\overline{\alpha}_i).c_i)_i) : P \mid \Delta}~.
$$
Here variables $\overline{x}_i$ and $\overline{\alpha}_i$ are bound in $c_i$
for each $i$.
Similarly we obtain a corresponding trace $L'$ for the superdeduction
inference introducing $P$ on the left which is turned into the typing rule
$$
\inferrule*[left=rL, right=$C'$]
{(c'_j \typed \Gamma, \overline{y}_j:\Gamma'_j \tez \overline{\beta}_j:\Delta'_j, \Delta)_j}
{\Gamma \mid r(L',(\tmu_j(\overline{y}_j,\overline{\beta}_j).c'_j)_j) : P \tez \Delta}~.
$$
Here variables $\overline{y}_j$ and $\overline{\beta}_j$ are bound in $c'_j$
for each $j$.
For example, the inference rules for $\subseteq$ are turned into
$$
\inferrule*[right=$\fox \notin \fv(\Gamma\comma\Delta)$]
{c \typed \Gamma, x: \fox \in A \tez \alpha:\fox \in B, \Delta}
{\Gamma \tez r(\fox?, \mu (x,\alpha). c) : A \subseteq B, \Delta}
\hspace{1em}\mbox{and}\hspace{1em}
\inferrule*[]
{c_1 \typed \Gamma, x : \fot \in B \tez \Delta \\
c_2 \typed \Gamma \tez \alpha: \fot \in A, \Delta}
{\Gamma, r(\fot!, \tmu_1 (x).c_1, \tmu_2 (\alpha).c_2) :A \subseteq B \tez \Delta}~.
$$
If $\mathcal{R}$ is a set of proposition rewrite rules, the type system
resulting of extending the type system of Figure \ref{subfig:lmmtTyping} with
the typing rules for $\mathcal{R}$ is denoted $\lmmtR$.

We must now define how cuts of the form
$$\langle~
r(L,(\mu_i(\overline{x}_i,\overline{\alpha}_i).c_i)_i) ~|~
r(L',(\tmu_j(\overline{y}_j,\overline{\beta}_j).c'_j)_j) ~\rangle$$
are reduced. Such reductions are computed using \emph{open $\lmmt$}, a type
system for derivations with open leaves\footnote{\emph{i.e.} leaves that remain unproven} in the $\lmmt$-calculus type system. An open leaf is
represented by a \emph{variable command} (symbols $X,Y\dots$). The types of
such variables have the same shape as the types of usual commands in
$\lmmt$-calculus: full sequents $\Gamma \tez \Delta$. Therefore typing in
open $\lmmt$ is performed in a context $\Theta$ which contains a list of typed
variable commands of the form $X \typed \Gamma \tez \Delta$. As usual,
variable commands are allowed to appear only once in such contexts. Typing
judgements are denoted
$$
\begin{array}{r@{\qquad}l@{\qquad}l}
& \Theta \Vdash c \typed \Gamma \tez \Delta & \mbox{when typing a command;} \\
& \Theta \Vdash \Gamma \tez \pi : A \mid \Delta & \mbox{when typing a term} \\
\mbox{and} & \Theta \Vdash \Gamma \mid e : A \tez \Delta & \mbox{when typing an environment.}
\end{array}
$$
Open $\lmmt$ is obtained by extending cut-free $\lmmt$ to such
judgements and by adding the typing rule
$$
\inferrule*[left=Open]
{ }
{\Theta; X \typed S \Vdash X \typed S}~.
$$
For example, Figure \ref{fig:example} contains a derivation of
$$
\begin{array}{l}
X \typed x:C \tez \alpha:D~ ; ~
Y \typed \tez \alpha:D, \beta:B \Vdash
\langle \lambda y. \mu \alpha. \langle y | (\mu \beta . Y) \cdot (\tmu x. X)
\rangle | \gamma \rangle
~\typed ~ (\tez \gamma : (B \imply C) \imply D)~.
\end{array}
$$
(where the prefix 
$
X \typed x:C \tez \alpha:D~;~
Y \typed \tez \alpha:D, \beta:B \Vdash
$ is omitted for readability.)
\begin{figure}
$$
\inferrule*[]
{
	\inferrule*[]
	{
		\inferrule*[]
		{
			\inferrule*[]
			{
				\inferrule*[leftskip=1.5em, rightskip=1.5em]
				{
					\inferrule*[]
					{\inferrule*[Left=Open]{ }{Y\typed \tez \beta : B, \alpha : D}}
					{\tez \mu \beta . Y : B \mid \alpha :D} \\
					\hspace{3em}
					\inferrule*[]
					{\inferrule*[Left=Open]{ }{X \typed x : C \tez \alpha :D}}
					{\mid \tmu x. X : C \tez \alpha :D}
				}
				{\mid (\mu \beta . Y) \cdot (\tmu x. X) : (B
\imply C) \tez \alpha  : D}
			}
			{\langle y | (\mu \beta . Y) \cdot (\tmu x. X) \rangle \typed y : (B
\imply C) \tez \alpha  : D}
		}
		{y : (B \imply C) \tez \mu \alpha. \langle y | (\mu \beta . Y) \cdot (\tmu x. X)
\rangle : D \mid}
	}
	{\tez \lambda y. \mu \alpha. \langle y | (\mu \beta . Y) \cdot (\tmu x. X)
\rangle : (B \imply C) \imply D \mid}
}
{\langle \lambda y. \mu \alpha. \langle y | (\mu \beta . Y) \cdot (\tmu x. X) \rangle | \gamma \rangle \typed \tez \gamma : (B \imply C) \imply D}
$$
\caption{Typing in open $\lmmt$}\label{fig:example}
\end{figure}

The reduction in Figure \ref{subfig:lmmtReductions} is extended to open
$\lmmt$ by simply defining how subtitutions behave on command variables ($X[t
\slash x]$, $X[e \slash \alpha]$ or $X[\fot \slash \fox]$): they are turned
into \emph{delayed substitutions}, \emph{i.e.} syntactic constructions,
denoted $X\{t \slash x\}$, $X\{e \slash \alpha\}$ or $X\{\fot \slash \fox\}$,
which will be turned back into primitive substitutions once $X$ is
instanciated.

A typing derivation in open $\lmmt$ obviously corresponds to a derivation in
LK (with open leaves). If $K$ is a typed command, term or environment, then
a trace for K is a trace for the derivation corresponding to $K$.
Let us reconsider our extended terms
$$r(L,(\mu_i(\overline{x}_i,\overline{\alpha}_i).c_i)_i)\qquad\mbox{ and }\qquad
r(L',(\tmu_j(\overline{x}_j,\overline{\alpha}_j).c'_j)_j)$$ and their respective typing
rules rR and rL. The sets
$$
\widetilde{\mbox{rR}} = 
\left\{ \pi \quad \Big\slash \quad
\begin{array}{l}
(X_i \typed (\overline{x}_i : \Gamma_i \tez \overline{\alpha}_i : \Delta_i))_i \Vdash
~\tez \pi : \varphi \mbox{ well-typed in open $\lmmt$ } \\
\mbox{and }L \mbox{ is a trace for } \pi
\end{array}
\right\}
$$
and
$$
\widetilde{\mbox{rL}} = 
\left\{ e \quad \Big\slash \quad
\begin{array}{l}
(Y_j \typed (\overline{y}_j : \Gamma'_j \tez \overline{\beta}_j : \Delta'_j))_j \Vdash e
: \varphi \tez \mbox{ well-typed in open $\lmmt$ } \\
\mbox{and }L'\mbox{ is a trace for } e
\end{array}
\right\}
$$
are both non-empty:
Indeed by construction of the superdeduction inference rules, we know that
there exists a derivation in LK of $\tez \varphi$ (resp.  $\varphi \tez$) from
premisses $(\Gamma_i \tez \Delta_i)_i$ (resp. $(\Gamma'_j \tez \Delta'_j)_j$)
such that $L$ (resp. $L'$) is a trace for this derivation. Therefore by logical
completeness of (open) $\lmmt$, there exists at least one term in
$\widetilde{\mbox{rR}}$ (resp. one environment in $\widetilde{\mbox{rL}}$).
Each term $\pi \in \widetilde{\mbox{rR}}$ intuitively represents
$r(L,(\mu_i(\overline{x}_i,\overline{\alpha}_i).X_i)_i)$ in open $\lmmt$. Each
environment $e \in \widetilde{\mbox{rL}}$ intuitively represents
$r(L',(\tmu_j(\overline{x}_j,\overline{\alpha}_j).c'_j)_j)$ in open $\lmmt$.
Therefore whenever $\pi$ and $e$ are respectively in $\widetilde{\mbox{rR}}$
and $\widetilde{\mbox{rL}}$, any normal form $\langle \pi | e \rangle$ can be
chosen as a direct reduct of
$$
\langle \quad r(L,(\mu_i(\overline{x}_i,\overline{\alpha}_i).c_i)_i)\quad | \quad
r(L',(\tmu_j(\overline{x}_j,\overline{\alpha}_j).c'_j)_j) \quad \rangle~.
$$
We suppose that for each typing rule rR (resp. rL) one
specific $\pi \in \widetilde{\mbox{rR}}$ (resp. one specific $e \in
\widetilde{\mbox{rL}}$) is distinguished. This term (resp. this environment)
is called the \emph{interpretation} of
$r(L,(\mu_i(\overline{x}_i,\overline{\alpha}_i).X_i)_i)$ (resp.
$r(L',(\tmu_j(\overline{x}_j,\overline{\alpha}_j).Y_j)_j)$).
Then for each normal form $c$ of $\langle \pi | e \rangle$, the rule
$$
\langle (
r(L,(\mu_i(\overline{x}_i,\overline{\alpha}_i).c_i)_i) |
r(L',(\tmu_j(\overline{y}_j,\overline{\beta}_j).c'_j)_j) \rangle \quad \rightarrow \quad
c[(c_i \slash X_i)_i,(c'_j \slash Y_j)_j]
$$
is added to the cut-elimination reduction (delayed
substitutions $\{\cdot \slash \cdot \}$ are replaced in $c$ by primitive
substitutions $[\cdot \slash \cdot]$).

Let us reconsider the inclusion example. The term $\pi = \lambda \fox .\lambda
x.  \mu\alpha. X$ is a potential interpretation of $r(\fox?, \mu (x,\alpha).
c)$. Indeed
$$X \typed x : \fox \in A \tez \fox \in B \Vdash \tez \pi : \forall
\fox. \fox \in A \imply \fox \in B \mid$$
is well-typed in open $\lmmt$ as demonstrated in Figure
\ref{subfig:openlmmtinclusionright} and $\fox?$ is a trace for $\pi$.
The environment $e = \fot \cdot (\mu \beta. Y) \cdot (\tmu y. Z)$ is a
potential interpretation of $r(\fot!, \tmu_1 (y).c_1, \tmu_2 (\beta).c_2)$.
Indeed
$$Y \typed \tez \beta : \fot \in A, Z \typed y : \fot \in B \Vdash \mid e :
\forall \fox. \fox \in A \imply \fox \in B \tez$$
is well-typed in open $\lmmt$ as demonstrated in Figure
\ref{subfig:openlmmtinclusionleft} and $\fot!$ is a trace for $e$.
\begin{figure}[t]
\centering
\subfigure[Typing $\lambda \fox .\lambda x.  \mu\alpha. X$]{
\centering
\begin{tabular}{c}
$
\inferrule*[]
{
	\inferrule*[]
	{
		\inferrule*[]
		{
			\inferrule*[]
			{
				\inferrule*[Left=Open]
				{ }
				{X \typed x : \fox \in A \tez \fox \in B \Vdash X \typed x : \fox \in A \tez
					\alpha : \fox \in B}
			}
			{X \typed x : \fox \in A \tez \fox \in B \Vdash x : \fox \in A \tez
			\mu\alpha. X : \fox \in B \mid}
		}
		{X \typed x : \fox \in A \tez \fox \in B \Vdash x : \fox \in A \tez
			\mu\alpha. X : \fox \in B \mid}
	}
	{X \typed x : \fox \in A \tez \fox \in B \Vdash \tez \lambda x.
		\mu\alpha. X : \fox \in A \imply \fox \in B \mid}
}
{X \typed x : \fox \in A \tez \fox \in B \Vdash \tez \lambda \fox .\lambda x.
\mu\alpha. X : \forall \fox. \fox \in A \imply \fox \in B \mid}
$
\\[2mm]
\end{tabular}
\label{subfig:openlmmtinclusionright}
}
\subfigure[Typing $\fot \cdot (\mu \beta. Y) \cdot (\tmu y. Z)$]{
\centering
\begin{tabular}{c}
$
\inferrule*[]
{
	\inferrule*[]
	{
		\inferrule*[vdots=4em, leftskip=5em, rightskip=25em]
		{
			\inferrule*[Left=Open]
			{ }
			{Y \typed \tez \beta : \fot \in A; Z \typed y : \fot \in B \Vdash Y
				\typed \tez \beta : \fot \in A}
		}
		{Y \typed \tez \beta : \fot \in A; Z \typed y : \fot \in B \Vdash \tez \mu
			\beta. Y : \fot \in A \mid} \quad \\
		\inferrule*[]
		{
			\inferrule*[Left=Open]
			{ }
			{Y \typed \tez \beta : \fot \in A; Z \typed y : \fot \in B \Vdash Z
				\typed y : \fot \in B \tez}
		}
		{Y \typed \tez \beta : \fot \in A; Z \typed y : \fot \in B \Vdash \mid
			\tmu y. Z : \fot \in B \tez}
	}
	{Y \typed \tez \beta : \fot \in A; Z \typed y : \fot \in B \Vdash \mid (\mu
		\beta. Y) \cdot (\tmu y. Z) : \fot \in A \imply \fot \in B \tez}
}
{Y \typed \tez \beta : \fot \in A, Z \typed y : \fot \in B \Vdash \mid \fot
	\cdot (\mu \beta. Y) \cdot (\tmu y. Z) : \forall \fox. \fox \in A \imply
	\fox \in B \tez}
$
\\[2mm]
\end{tabular}
\label{subfig:openlmmtinclusionleft}}
\caption{Typing interpretations for inclusion}
\end{figure}
The cut
$$
\langle \lambda \fox .\lambda x.  \mu\alpha. X| \fot \cdot (\mu \beta. Y) \cdot
(\tmu y. Z)\rangle
$$
has two normal forms, namely
$$
X\{\fot \slash \fox\}\{(\mu\beta. Y)\slash x\}\{\tmu y. Z \slash \alpha\}
\qquad\mbox{and}\qquad
Z\{\mu \alpha. X\{\fot \slash \fox\}\{(\mu\beta. Y)\slash x\} \slash y\}~.
$$
Therefore a cut
$$
\langle r(\fox?, \mu (x,\alpha).  c)| r(\fot!, \tmu_1 (y).c_1, \tmu_2
(\beta).c_2)\rangle
$$
reduces to
$$
c[\fot \slash \fox][(\mu\beta. c_2)\slash x][\tmu y. c_1 \slash \alpha]
\qquad\mbox{and}\qquad
c_1[\mu \alpha. c[\fot \slash \fox][(\mu\beta. c_2)\slash x] \slash y]~.
$$

If $\mathcal{R}$ is a set of proposition rewrite rules, the reduction relation
of Figure \ref{subfig:lmmtReductions} extended by the reduction rules for
$\mathcal{R}$ will be denoted $\rightarrow_{\lmmtR}$.

\begin{theo}[Subject Reduction]\label{theo:sr}
For all $\mathcal{R}$, typability in $\lmmtR$ is preserved by reduction
through $\rightarrow_{\lmmtR}$.
\end{theo}

\begin{proof}
The only case worth considering is a reduction of some \emph{supercut}
$$
\langle r(L,(\mu_i(\overline{x}_i,\overline{\alpha}_i).c_i)_i) |
r(L',(\tmu_j(\overline{y}_j,\overline{\beta}_j).c'_j)_j) \rangle~.
$$
If $\pi$ and $e$ are the respective interpretations of
$r(L,(\mu_i(\overline{x}_i,\overline{\alpha}_i).X_i)_i)$ and
$r(L',(\tmu_j(\overline{y}_j,\overline{\beta}_j).Y_j)_j)$ and $c$ is
a normal form of $\langle \pi | e \rangle$, then the supercut reduces to
$c[(c_i \slash
X_i)_i,(c'_j \slash Y_j)_j]$.
By definition of the interpretations, the judgements
$(X_i \typed (\overline{x}_i : \Gamma_i \tez \overline{\alpha}_i :
\Delta_i))_i \Vdash\,\tez \pi : \varphi$
and
$(Y_j \typed (\overline{y}_j : \Gamma'_j \tez \overline{\beta}_j :
\Delta'_j))_j \Vdash e : \varphi \tez$ are well-typed in open $\lmmt$.
Therefore by subject reduction in open $\lmmt$
$$
(X_i \typed (\overline{x}_i : \Gamma_i \tez \overline{\alpha}_i :
\Delta_i))_i \, ; \, (Y_j \typed (\overline{y}_j : \Gamma'_j \tez \overline{\beta}_j :
\Delta'_j))_j \Vdash c \, \typed \, \tez
$$
is also well-typed in open $\lmmt$.
Then a simple substitution lemma on command variables\footnote{not detailed
here for simplicity} proves that if the command
$$\langle r(L,(\mu_i(\overline{x}_i,\overline{\alpha}_i).c_i)_i) |
r(L',(\tmu_j(\overline{y}_j,\overline{\beta}_j).c'_j)_j) \rangle$$ has a
certain type, then so does the command $c[(c_i \slash X_i)_i,(c'_j \slash Y_j)_j]$.
\end{proof}

\begin{theo}[Strong Normalisation]
For all $\mathcal{R}$ satisfying hypothesis \ref{hypo},
$\rightarrow_{\lmmtR}$ is strongly normalising on commands, terms and
environments that are well-typed in $\lmmtR$.
\end{theo}

\begin{proof}
Hypothesis \ref{hypo} implies that any formula $\varphi$ has a unique normal
form for $\mathcal{R}$ that we denote $\varphi \nfp$. Let us denote
$\rightarrow_e$ the rewrite relation defined by replacing extended terms for
superdeduction by their interpretations.
$$
\begin{array}{l@{~\rightarrow_e~}l}
r(L,(\mu_i(\overline{x}_i,\overline{\alpha}_i).c_i)_i) & \pi[c_i\slash X_i] \\
r(L',(\tmu_j(\overline{y}_j,\overline{\beta}_j).c'_j)_j) & e[c'_j\slash Y_j] \\
\multicolumn{2}{c}{\dots}
\end{array}
$$
Such a rewrite relation is strongly normalising and confluent, therefore
yielding for any extended command $c$, term $\pi$ or environment $e$ a normal
form denoted $c \nfe$, $\pi \nfe$ of $e \nfe$. Such normal forms are raw
$\lmmt$ commands, terms or environments. Strong normalisation of our extended
cut-elimination reduction comes from the facts that
\begin{inparaenum}[\bf 1.]
\item $c \typed \Gamma \tez \Delta$ well-typed in our extended type system
implies that $c \nfe \typed (\Gamma) \nfp \tez (\Delta) \nfp$ well-typed in
$\lmmt$~;
\item $\Gamma \tez \pi : A \mid \Delta$ well-typed in our extended type system
implies that $(\Gamma) \nfp \tez \pi \nfe : A \nfp \mid (\Delta) \nfp$
well-typed in $\lmmt$~;
\item $\Gamma \mid e : A \tez \Delta$ well-typed in our extended type system
implies that $(\Gamma) \nfp \mid e \nfe : A \nfp \tez (\Delta) \nfp$
well-typed in $\lmmt$~;
\item $c \rightarrow c'$ implies $c \nfe \rightarrow^+ c' \nfe$
\item $\pi \rightarrow \pi'$ implies $\pi \nfe \rightarrow^+ \pi' \nfe$
\item $e \rightarrow e'$ implies $e \nfe \rightarrow^+ e' \nfe$.
\end{inparaenum}
The hypothesis on first-order function symbols (see Hypothesis \ref{hypo}) is
crucial in establishing points 1 to 3: indeed for any formula $\varphi$ and
any first-order substitution $\sigma$, it must be the case that $(\varphi
\nfp)\sigma = (\varphi\sigma)\nfp$.
These six points (combined with Theorem \ref{theo:sr}) demonstrate that
through $\nfe$ and $\nfp$, the $\lmmt$-calculus simulates our extended
calculus: any well-typed reduction in our extended calculus induces through
$\nfe$ and $\nfp$ a longer well-typed reduction in $\lmmt$. Strong
normalisation of $\lmmt$ therefore implies strong normalisation of our
extended reduction.
\end{proof}

The end of this section is dedicated to a pathological example for
superdeduction: the proposition rewrite rule
$$
r~:~P \rightarrow (\exists \fox_1. \forall \fox_2.
A(\fox_1,\fox_2 )) \vee (\exists \foy_1. \forall \foy_2. B(\foy_1, \foy_2))
$$
whose \emph{most general} superdeduction rules are
$$
\inferrule*[right=$\left\{\begin{array}{l} \fox_2 \notin
\fv(\Gamma\comma\Delta\comma \fou ) \\ \foy_2 \notin
\fv(\Gamma\comma\Delta)\end{array}\right.$]
{\Gamma \tez A(\fot, \fox_2), B(\fou, \foy_2), \Delta}
{\Gamma \tez P, \Delta}
\quad\mbox{and}\quad
\inferrule*[right=$\left\{\begin{array}{l} \fox_2 \notin
\fv(\Gamma\comma\Delta) \\ \foy_2 \notin
\fv(\Gamma\comma\Delta\comma \fot )\end{array}\right.$]
{\Gamma \tez A(\fot, \fox_2), B(\fou, \foy_2), \Delta}
{\Gamma \tez P, \Delta}~.
$$
The original proofterm extension \cite{BraunerHK07} transforms these two
inferences into a unique proofterm\linebreak $rR(\lambda \fox_2. \lambda \foy_2.
(\lambda \alpha. \lambda \beta.  m), \fot, \fou,\gamma)$. It is obviously
inaccurate with respect to the scope of $\fox_2$ and $\foy_2$: in the
proofterm there is no mention that either $\fot$ is not in the scope of
$\foy_2$ or $\fou$ is not in the scope of $\fox_2$. This fact is not reflected
in the pure syntax but in the typing rules
$$
\inferrule*[right=$\left\{\begin{array}{l} \fox_2 \notin
\fv(\Gamma\comma\Delta\comma \fou ) \\ \foy_2 \notin
\fv(\Gamma\comma\Delta)\end{array}\right.$]
{m\typed\Gamma \tez \alpha:A(\fot, \fox_2), \beta:B(\fou, \foy_2), \Delta}
{rR(\lambda \fox_2. \lambda \foy_2. (\lambda \alpha. \lambda \beta. m), \fot,
\fou,\gamma)\typed\Gamma \tez \gamma:P, \Delta}
$$
and
$$
\inferrule*[right=$\left\{\begin{array}{l} \fox_2 \notin
\fv(\Gamma\comma\Delta) \\ \foy_2 \notin
\fv(\Gamma\comma\Delta\comma \fot )\end{array}\right.$]
{m\typed\Gamma \tez \alpha:A(\fot, \fox_2), \beta:B(\fou, \foy_2), \Delta}
{rR(\lambda \fox_2. \lambda \foy_2. (\lambda \alpha. \lambda \beta. m), \fot,
\fou,\gamma)\typed\Gamma \tez \gamma:P, \Delta}~.
$$

Let us see how this mistake is corrected in our extension of the
$\lmmt$-calculus. Traces for the superdeduction inferences are respectively
$\fou!\,\foy_2?\,\fot!\,\fox_2?$ and
$\fot!\,\fox_2?\,\fou!\,\foy_2?$. These traces clearly specify that
whether $\fot$ is not in the scope of $\foy_2$ or $\fou$ is not in the scope of
$\fox_2$. Our extension of the $\lmmt$-calculus translates these superdeduction
inferences into the typing rules
$$
\inferrule*[right=$\left\{\begin{array}{l} \fox_2 \notin
\fv(\Gamma\comma\Delta\comma \fou ) \\ \foy_2 \notin
\fv(\Gamma\comma\Delta)\end{array}\right.$]
{c\typed\Gamma \tez \alpha:A(\fot, \fox_2), \beta:B(\fou, \foy_2), \Delta}
{\Gamma \tez r(\fou!\,\foy_2?\,\fot!\,\fox_2? ,\mu(\alpha,\beta).c):P \mid \Delta}
$$
and
$$
\inferrule*[right=$\left\{\begin{array}{l} \fox_2 \notin
\fv(\Gamma\comma\Delta) \\ \foy_2 \notin
\fv(\Gamma\comma\Delta\comma \fot )\end{array}\right.$]
{c\typed\Gamma \tez \alpha:A(\fot, \fox_2), \beta:B(\fou, \foy_2), \Delta}
{\Gamma \tez r(\fot!\,\fox_2?\,\fou!\,\foy_2? ,\mu(\alpha,\beta).c):P\mid
\Delta}~.
$$
The proofterms (and the typing rules) reflect the scope of the
eigenvariables. The \emph{interpretation} of
$r(\fou!\,\foy_2?\,\fot!\,\fox_2? ,\mu(\alpha,\beta).c)$ is by definition a
term well-typed in $\lmmt$ whose trace is $\fou!\,\foy_2?\,\fot!\,\fox_2?$
and the \emph{interpretation} of $r(\fot!\,\fox_2?\,\fou!\,\foy_2?
,\mu(\alpha,\beta).c)$ is by definition a term well-typed in $\lmmt$ whose
trace is $\fot!\,\fox_2?\,\fou!\,\foy_2?$. This trace restriction implies that
$r(\fou!\,\foy_2?\,\fot!\,\fox_2? ,\mu(\alpha,\beta).c)$ and
$r(\fot!\,\fox_2?\,\fou!\,\foy_2? ,\mu(\alpha,\beta).c)$ behave differently
with respect to cut-elimination.

\section{Conclusion}

This extension of the $\lmmt$-calculus is a first step towards a
computational interpretation of superdeduction. Indeed it refutes the idea
\cite{BraunerHK07} that Christian Urban's calculus is a better basis for a
proofterm language for superdeduction: $\lmmt$ syntax, typing and reduction is
as suitable as Christian Urban's calculus for superdeduction. The
extension presented in this short paper is almost a mechanical transcription of
the original extension \cite{BraunerHK07}. It relates superdeduction more
closely to the $\lambda$-calculus based Curry-Howard correspondence without
exploring any further the computational content of cut-elimination for
superdeduction.

We believe that one of the key ingredients towards this goal is
pattern-matching. Indeed superdeduction systems historically come from
\emph{supernatural deduction} \cite{WackThese05}, an extension of natural
deduction designed to type the rewriting-calculus (\emph{a.k.a.}
$\rho$-calculus) \cite{rhoCalIGLP-I+II-2001}. Supernatural deduction turns
proposition rewrite rules of the form
$$
r~:~
P \quad \rightarrow \quad \forall \bar{\fox}. ((A_1 \land A_2 \dots A_n) \imply C)
$$
into inference rules for natural deduction
$$
\inferrule*[right=$\bar{\fox} \notin \fv(\Gamma)$]
{\Gamma, A_1\dots A_n \tez C}
{\Gamma \tez P}
\qquad\mbox{and}\qquad
\inferrule*[]
{\Gamma \tez P \\ (\Gamma \tez A_i[\bar{\fot}\slash\bar{\fox}])_i}
{\Gamma \tez C[\bar{\fot}\slash\bar{\fox}]}~.
$$
(The first rule is an introduction rule and the second is an elimination
rule.) The rewriting calculus is an extension of the $\lambda$-calculus where
rewrite rules replace lambda-abstractions.
The idea underlying the relation between supernatural deduction and rewriting
calculus is that the proposition rewrite rule $r$ corresponds to a specific
pattern $r(\bar{\fox},x_1\dots x_n)$. The introduction rule types an abstraction
on this pattern (\emph{i.e.} a rewrite rule)
$$
\inferrule*[right=$\bar{\fox} \notin\fv(\Gamma)$]
{\Gamma, x_1:A_1 \dots x_n:A_n \tez \pi : C}
{\Gamma \tez r(\bar{\fox},x_1\dots x_n) \rightarrow \pi : P}~.
$$
Dually the elimination rule types an application on this pattern
$$
\inferrule*[]
{\Gamma \tez \pi : P \\ (\Gamma \tez \pi_i : A_i[\bar{\fot}\slash
\bar{\fox}])_i}
{\Gamma \tez \pi \, r(\bar{\fot},\pi_1\dots \pi_n) : C}~.
$$
Supernatural deduction systems (in intuitionistic natural deduction) have
later been transformed into superdeduction systems (in classical sequent
calculus) in order to handle more general proposition rewrite rules.
This transformation from supernatural deduction to superdeduction systems
should not break the relation with pattern matching. Indeed cut-elimination in
sequent calculus relates to pattern matching \cite{CerritoK04}. Recent
analysis shows that the duality between patterns and terms reflects the duality
between phases in focused proof systems \cite{Zeilberger08}. Finally we
demonstrated \cite{Houtmann08,TheseClement} that superdeduction systems share
strong similarities with focused proof systems such as LKF
\cite{LiangM07,LiangM09}, a focused sequent calculus for classical logic.
Answers should naturally arise from the study of the computational content of
such focused systems \cite{Munch-Maccagnoni09,CurienM10}.

\bibliographystyle{eptcs}
\bibliography{paper}

\end{document}